\documentclass[letterpaper, 10 pt, conference]{ieeeconf}  
\IEEEoverridecommandlockouts
\usepackage{amsmath}
\usepackage{amsfonts}
\usepackage{graphicx}
\usepackage{xcolor}
\usepackage{figsize}
\usepackage{upgreek}
\usepackage{amssymb}
\usepackage{algorithmic}
\usepackage[linesnumbered,ruled,vlined]{algorithm2e}
\usepackage{bm}

\newtheorem{definition}{\bf Definition}
\newtheorem{theorem}{\bf Theorem}

\newtheorem{assumption}{\bf Assumption}

\newtheorem{remark}{\bf Remark}
\newtheorem{notation}{Notation}

\newcommand\norm[1]{\left\lVert#1\right\rVert}

\title{\LARGE \bf
Observer-Based Environment Robust Control Barrier Functions for Safety-critical Control with Dynamic Obstacles}

\author{Ying Shuai Quan$^{1}$, Jian Zhou$^{2}$, Erik Frisk$^{2}$, Chung Choo Chung$^{1^\dag}$ 
\thanks{$^{1}$Y. S. Quan is with Dept. of Electrical Engineering, Chalmers University of Technology, 412 96 Göteborg, Sweden.
        {\tt\small {yeongsu.quan@gmail.com}}}
\thanks{$^{2}$J. Zhou and E. Frisk are with the Division of Vehicular Systems, Dept. of Electrical Engineering, Linköping  University, Linköping, Sweden.
        {\tt\small \{jian.zhou, erik.frisk\}@liu.se}}
\thanks{$^{1^\dag}$C. C. Chung is with Dept. of Electrical Engineering, Hanyang University, Seoul 04763, Republic of Korea.
        {\tt\small {cchung@hanyang.ac.kr}}}
}
\begin{document}

\maketitle
\thispagestyle{empty}
\pagestyle{empty}

\begin{abstract}
    This paper proposes a safety-critical controller for dynamic and uncertain environments, leveraging a robust environment control barrier function (ECBF) to enhance the robustness against the measurement and prediction uncertainties associated with moving obstacles. The approach reduces conservatism, compared with a worst-case uncertainty approach, by incorporating a state observer for obstacles into the ECBF design.
    The controller, which guarantees safety, is achieved through solving a quadratic programming problem. The proposed method's effectiveness is demonstrated via a dynamic obstacle-avoidance problem for an autonomous vehicle, including comparisons with established baseline approaches.
\end{abstract}
\section{Introduction}
Control barrier functions (CBFs)-combined safety-critical controllers have been demonstrated as promising approaches for guaranteeing the safety of control systems in numerous applications~\cite{ames2016control,ames2019control}. In the controller design, CBFs can be used as a safety constraint in an optimization problem to find the control actions subject to a safe set. One of the main challenges in designing a CBF-based safety controller is ensuring strict safety in dynamic environments in the presence of moving surrounding obstacles~\cite{molnar2022safety,hamdipoor2023safe}. This necessitates a safety-critical controller that accounts for the interactions between the control system and its environment, ensuring collision avoidance with dynamic obstacles.

Recent research in the literature has introduced the concept of time-varying CBFs, incorporating the time derivative of the CBF into the constraint of the controller to account for the influence of dynamic environments~\cite{igarashi2019time,chalaki2022barrier,he2021rule,wu2016safety}. These studies mostly assume the availability of perfect environmental information. However, in practical scenarios, the accurate perception of the environmental state is challenging. The environmental information that is reconstructed through sensor measurements usually contains uncertainties, which can potentially result in unsafe behavior of the control system. Therefore, the development of a robust CBF for dynamic environments is important to guarantee the safety of the system in the presence of uncertain moving obstacles.

There are a few studies addressing robust CBF formulations for uncertain dynamic environments. A notable contribution in this domain is found in~\cite{molnar2022safety}, which introduced the concept of environment CBFs (ECBFs) and explored robust ECBFs to mitigate errors induced by time delays when measuring the states of the environment. The design of the robust ECBF was achieved by solving a second-order cone programming (SOCP) problem. To remedy the complexities of solving the problem associated with the SOCP, \cite{hamdipoor2023safe} proposed an alternative robust ECBF approach based on solving a quadratic programming (QP) problem. This method aims to minimize the deviations of the system's control input with a reference input obtained by solving a QP problem of a nominal ECBF.  In both proposed strategies, the designed robust ECBFs considered the worst-case errors of observations of environmental states to enhance the system's robustness against uncertainties in dynamic environments.

Despite the guaranteed robustness achieved through accounting for worst-case observation errors in ECBF design, such approaches can be overly conservative, potentially limiting the efficiency and feasibility of the controller. To address these concerns, this paper proposes a novel observer-based robust ECBF algorithm. This methodology aims to mitigate conservatism by not directly utilizing environmental measurements with the worst-case error, but instead estimating the environmental states through a bounded-error observer. Based on the observed states of the environment, the control input of the system is efficiently obtained by solving a QP problem. The contributions of the research are stated below:
\begin{itemize}
\item An observer-based environment robust ECBF is developed for strict safety guarantee in the presence of dynamic obstacles with bounded measurement disturbances. The method achieves robustness by considering the measurement uncertainties and reduces conservatism by leveraging the observations of the obstacle. 
\item The controller computes the input for the control system by solving a QP problem, allowing for a low-computational complexity and efficient solution.
\item The effectiveness of the methods is demonstrated in simulations for safe control of an autonomous ego vehicle with a dynamic and uncertain obstacle, by comparisons with both nominal ECBF and robust ECBF.
\end{itemize}

\begin{notation}
A continuous function $\alpha:[0,a) \rightarrow [0,\infty)$ for some $a>0$ is said to belong to class $\mathcal{K}$ if it is strictly increasing and $\alpha(0) = 0$.
The Euclidean norm of a vector $x$ is denoted as $\norm{x} = \sqrt{x^T x}$, $\mathbb{R}^n$ and $\mathbb{R}_+^n$ mean the $n$-dimensional real number vector space and $n$-dimensional non-negative real number vector space, respectively. An interval of integers is denoted by $\mathbb{I}_a^b = \{a, a+1, \cdots, b\}$. $I^n$ indicates an $n \times n$ identity matrix. A matrix of appropriate dimension with all elements equal to $0$ is denoted by $\bm{0}$. A matrix $\bm{0} \preceq A$ means it is semi-positive definite. The Lie derivative of a scalar function $h: \mathcal{X}\rightarrow\mathbb{R},  (\mathcal{X} \subset \mathbb{R}^n)$, along a vector field $f: \mathcal{X}\rightarrow\mathbb{R}^n$ is denoted $L_fh(x) = \frac{\partial h}{\partial x}f(x)$.
\end{notation}

\section{Preliminaries}\label{sec:preliminaries}
This section introduces the primary concepts of CBF and ECBF, and the controllers based on the CBFs, providing a foundation for the proposed method in Section~\ref{sec:robust ECBF}.
\subsection{Control Barrier Function (CBF)}
Consider a continuous-time dynamic system with affine control inputs as:
\begin{equation}
\dot x(t) = f(x(t))+g(x(t))u(t),\ x(0) = x_0,
\label{eq:sys}
\end{equation}
where $t\geq 0$ is the time instant, $x(t)\in \mathcal{X} \subseteq \mathbb{R}^{n_x}$ is the state, $u(t)\in \mathcal{U} \subseteq \mathbb{R}^{n_u}$ is the control input, $\mathcal{X}$ and $\mathcal{U}$ are the set of feasible states and inputs, respectively. The functions $f: \mathcal{X} \rightarrow \mathbb{R}^{n_x}$ and $g: \mathcal{X} \rightarrow \mathbb{R}^{n_x\times n_u}$ are locally Lipschitz continuous on $\mathcal{X}$. We assume there exists a feedback controller $u=k(x)$ with $k: \mathcal{X} \rightarrow \mathcal{U}$ that is locally Lipschitz continuous, and the system $(\ref{eq:sys})$ has a unique solution for $t \geq 0$. In the rest of the paper, the time variable $t$ is omitted from $x(t)$ and $u(t)$ to simplify the notation, such that $x(t) := x$ and $u(t) := u$. This also applies to the variables defined in \eqref{eq:aug_sys} and \eqref{eq:obstacle state}.

We define a set $\mathcal{C}\subset \mathcal{X}$ representing the safe states of system $(\ref{eq:sys})$ as a zero-superlevel set of a continuously differentiable function $h: \mathcal{X} \rightarrow \mathbb{R}$, i.e.,
\begin{equation}\label{eq: C set}
\begin{split}
\mathcal{C} = \{  x  \in \mathcal{X} : h(x) \geq 0 \}.
\end{split}
\end{equation}

The safe state of~\eqref{eq:sys} is defined as the state $x$ remaining within the set $\mathcal{C}$ for all $t\geq 0$. Therefore, the objective of a safety-guaranteed controller of system~\eqref{eq:sys} is to render the set $\mathcal{C}$ forward invariant. This requires guaranteeing that $\forall x_0 \in \mathcal{C}$, the state satisfies $x\in\mathcal{C}$, $\forall t\geq 0$. This control objective can be efficiently achieved by leveraging the CBF~\cite{ames2016control}.

\begin{definition}
The function $h(x)$ is a CBF for system \eqref{eq:sys} if there exists a class $\mathcal{K}$ function $\alpha$, such that $\forall x\in\mathcal{C}$,
\begin{equation}\label{eq:h function}
\sup_{u\in \mathcal{U}}L_fh(x)+L_gh(x)u \geq -\alpha(h(x)).
\end{equation}
\end{definition}

Following the definition in~\eqref{eq:h function}, the safety guarantee of CBF is formally established in Theorem~\ref{theo:CBF definition}~\cite{ames2016control}.

\begin{theorem}\label{theo:CBF definition}
If $h(x)$ is a CBF for \eqref{eq:sys}, then any locally Lipshitz continous controller $u=k(x)$ satisfying 
\begin{equation}
L_fh(x)+L_gh(x)u \geq -\alpha(h(x)), \forall x \in \mathcal{C} \label{eq:CBF_condition}
\end{equation}
renders $\mathcal{C}$ forward invariant (safe), i.e., it ensures $\forall x_0\in \mathcal{C} \Rightarrow x \in \mathcal{C}, \forall t \geq 0$.
\end{theorem}
\begin{proof}
The proof can be found in~\cite{ames2016control, ames2019control}.
\end{proof}
\begin{remark}
Condition \eqref{eq:CBF_condition} is often used in designing optimization-based controllers. Given a desired control input $u_{\rm des} \in \mathcal{U}$, a quadratic programming (QP)-based safe controller using \eqref{eq:CBF_condition} can be formulated as
\begin{equation}
\begin{split}
k(x) &= \arg\min_{u \in \mathcal{U}} \norm{u-u_{\rm des}}^2\\
 &~~~~~~~~{\rm{s.t.}}~ L_fh(x)+L_gh(x)u \geq -\alpha(h(x)).
\end{split}
\label{eq:CBF_QP}
\end{equation}
\end{remark}
\subsection{Environmental Control Barrier Function (ECBF)}
In dynamic driving scenarios where the safe set $\mathcal{C}$ depends on the environment, the safety of system~\eqref{eq:sys} is affected by the time-varying environment state $e \in E \subseteq \mathbb{R}^{n_e}$ with $\dot e \in \mathcal{E} \subseteq \mathbb{R}^{n_e}$ and $e(0) = e_0$. The safety-critical controller in dynamic environments can be designed based on ECBF~\cite{molnar2022safety, hamdipoor2023safe}. An ECBF is a time-varying extension of CBF, where both $e$ and $\dot{e}$ are considered in the design of the safe controller. For the design of the ECBF, we first formulate an augmented system consisting of system state $x$ and the environment state $e$ as
\begin{equation}
    \dot\xi = F(\xi)+G(\xi)u_\xi,\ \xi(0) = \xi_0,
\label{eq:aug_sys}
\end{equation}
with
\begin{align*}
\begin{split}
    \xi = 
    \begin{bmatrix}
        x\\e
    \end{bmatrix},
    u_\xi = 
    \begin{bmatrix}
        u\\\dot e
    \end{bmatrix},
    F(\xi) = 
    \begin{bmatrix}
        f(x)\\0
    \end{bmatrix},
    G(\xi) = 
    \begin{bmatrix}
        g(x)\\I
    \end{bmatrix}.
\end{split}
\end{align*}

The safe set $\mathcal{C}_\xi$ corresponding to \eqref{eq:aug_sys} is then defined as
\begin{equation}
\mathcal{C}_\xi = \{  \xi  \in \mathcal{X}\times E : H(\xi) \geq 0 \},
\label{eq:C_xi}
\end{equation}
where $H: \mathcal{X} \times E \rightarrow \mathbb{R}$ is a continuously differentiable function. Definition~\ref{def:ECBF} provides the condition that $H(\xi)$ should satisfy as an ECBF.
\begin{definition}\label{def:ECBF}
The function $H(\xi)$ is an ECBF for (\ref{eq:aug_sys}) if there exists a class $\mathcal{K}$ function $\alpha$, such that $\forall \xi\in\mathcal{C_\xi}$ and $\forall \dot e \in\mathcal{E}$
\begin{equation}\label{eq:ECBF}
\sup_{u\in \mathcal{U}}L_F H(\xi)+L_G H(\xi)u_\xi > -\alpha(H(\xi)).
\end{equation}
\end{definition}

Following the definition of ECBF, the safety guarantee of the control system~\eqref{eq:sys} is shown in Theorem~\ref{theo:safety ECBF}~\cite{molnar2022safety}.
\begin{theorem}\label{theo:safety ECBF}
If $H(\xi)$ is an ECBF for (\ref{eq:aug_sys}), then any locally Lipshitz continuous controller $u=K(\xi,\dot e)$ satisfying 
\begin{equation}
L_F H(\xi)+L_G H(\xi)u_\xi \geq -\alpha(H(\xi)), \forall \xi \in \mathcal{C_\xi}
\label{eq:ECBF_condition}
\end{equation}
renders $\mathcal{C_\xi}$ forward invariant (safe), i.e., it ensures $\forall \xi_0\in \mathcal{C_\xi} \Rightarrow \xi \in \mathcal{C_\xi}, \forall t \geq 0$.
\end{theorem}
\begin{proof}
The proof can be found in~\cite{molnar2022safety}.
\end{proof}
\begin{remark}
Using the results of Theorem~\ref{theo:safety ECBF}, a QP-based safety-critical controller for the control system~\eqref{eq:sys} in dynamic environments can be formulated as:
\begin{equation}\small \label{eq:ECBF QP}
\begin{split}
    K(\xi,\dot e) &= \arg\min_{u \in \mathcal{U}} \norm{u-u_{\rm des}}^2\\
    &~~~~~~~~{\rm s.t.}~ L_F H(\xi)+L_G H(\xi)u_\xi \geq -\alpha(H(\xi)).
\end{split}
\end{equation}
where $u_{\rm des} \in \mathcal{U}$ is a desired input. Note that the decision variable of \eqref{eq:ECBF QP} is $u$ instead of $u_\xi$ of \eqref{eq:aug_sys}, as $\dot e$ is a known parameter for the controller.
\end{remark}
\section{Observer-Based Robust ECBF}\label{sec:robust ECBF}
The design of ECBF relies on the knowledge of the environment state $e$ and its derivative $\dot e$, which are usually estimated with uncertainties. Therefore, it becomes necessary to robustify the safe controller against measurement and prediction uncertainties of $e$ and $\dot{e}$. Instead of using worst-case uncertainty bounds of the environment states in ECBFs as in \cite{molnar2022safety, hamdipoor2023safe}, this paper proposes an observer-based robust ECBF to reduce the conservatism of the controller and maintain the robustness of the performance.

\subsection{Bounded-Error Observer for Surrounding Obstacles}\label{sec:BE observer theory}
In the case of obstacle avoidance, the environment state $e$ typically contains positional elements of the surrounding obstacle's state $x^s$. The state $x^s$ of the surrounding obstacle can be described by a general nonlinear system model as
\begin{equation}\label{eq:obstacle state}
\begin{split}
\dot x^s &= f^s(x^s, u^s),\ x^s(0) = x^s_0\\
y^s &=  c(x^s) +  c_w(x^s)w,
\end{split}
\end{equation}
where {$x^s\in \mathcal{X}^s \subseteq \mathbb{R}^{n_x^s}$} is the state of the surrounding obstacle, $u^s\in \mathcal{U}^s \subseteq \mathbb{R}^{n_u^s}$ is the control input, and $y^s\in  \mathbb{R}^{n_y^s}$ is the output. The variable $w: \mathbb{R}_{+} \rightarrow \mathbb{R}^{n_w}$ is the measurement noise, which is assumed to be piece-wise continuous and bounded with $\norm{w}_\infty \leq  \bar w$ for the known bound $\bar w < \infty$. The functions
$f: \mathcal{X}^s \times \mathcal{U}^s \rightarrow \mathcal{X}^s, c: \mathcal{X}^s \rightarrow \mathbb{R}^{n_y^s}$ and $c_w: \mathcal{X}^s \rightarrow \mathbb{R}^{n^s_y \times n^s_v}$ are locally Lipschitz continuous.

For system~\eqref{eq:obstacle state}, the observer that maintains a state estimate $\hat x^s \in \mathcal{X}^s$ can be formulated as
\begin{equation}
\dot{\hat x}^s = p(\hat x^s, y^s)+q(\hat x^s, y^s)u^s,\ \hat{x}^s(0) = {\hat x}^s_0,
\label{eq:observer}
\end{equation}
where functions $p: \mathcal{X}^s\times \mathbb{R}^{n_y^s}\rightarrow \mathbb{R}^{n_x^s}$, and $q: \mathcal{X}^s\times \mathbb{R}^{n_y^s}\rightarrow \mathbb{R}^{n_x^s\times n_u^s}$ are locally Lipschitz.
This paper applies a bounded-error (BE) observer, as defined in Definition~\ref{def:BE observer}, to guarantee the boundedness of the estimation error.
\begin{definition}\label{def:BE observer}
If an initial bounded set $\mathcal{D}(\hat x^s_0)\subset \mathcal{X}^s$ and a time-varying bounded set $\mathcal{P}(t,\hat x^s)\subset \mathcal{X}^s$ satisfy
\begin{equation}
x^s_0 \in \mathcal{D}(\hat x^s_0) 
\Rightarrow 
x^s\in\mathcal{P}(t,\hat x^s), \forall t \geq 0,
\label{eq:BE_con}
\end{equation}
then the observer \eqref{eq:observer} is a BE observer.
\end{definition}

Definition~\ref{def:BE observer} presents the condition of a general BE observer applicable in the proposed method. The design of the BE observer will be specified within the context of a particular case study in Section~\ref{sec: BE observer case study}. Following the BE observer, the estimation error of the environment state will be derived in Section~\ref{sec:estimation error}, and Section~\ref{sec:robust ECBF sub} designs the robust ECBF based on the estimation error.
\subsection{Estimation Errors}\label{sec:estimation error}
Denote by $\hat e$ and $\hat{\dot e}$ as the estimate of $e$ and $\dot e$, respectively, the estimation error of $e$ and $\dot{e}$ are defined as
\begin{equation}
d_{e} = e-\hat e,~
d_{\dot e} = \dot e - \hat{\dot e}.
\end{equation}
\begin{assumption}
These estimates have known bounds $\varepsilon_1\in\mathbb{R}_{+}$ and $\varepsilon_2\in\mathbb{R}_{+}$, i.e.,
\begin{equation}\label{eq:disturbance bound}
    \norm{d_e}\leq \varepsilon_1,~
\norm{d_{\dot e}}\leq \varepsilon_2.
\end{equation}
\end{assumption}

Based on the estimated state $\hat{x}^s$ and $\dot{\hat{x}}^s$ of the obstacle, the estimation of the environmental state in the design of the robust ECBF can be obtained through
\begin{equation}
\hat e = c_e(\hat x^s),\ \hat{\dot e} = c_{\dot e}(\hat {\dot x}^s),
\end{equation}
where functions $c_e$ and $c_{\dot e}$ are locally Lipschitz continuous with Lipschitz coefficients $\mathcal{L}_{c_e}$, and $\mathcal{L}_{c_{\dot e}}$.
Thus, the parameters $\varepsilon_1$ and $\varepsilon_2$ in \eqref{eq:disturbance bound} can be found as
\begin{equation}
\varepsilon_1 = \mathcal{L}_{c_e}\norm{x^s-\hat{x}^s},\
\varepsilon_2 = \mathcal{L}_{c_{\dot e}}\norm{x^s-\hat{x}^s},
\label{eq:e_bound}
\end{equation}
where $\norm{x^s-\hat{x}^s}$ is found via the time-varying bounded set $\mathcal{P}(t,\hat x^s)$ as defined in \eqref{eq:BE_con}. This will be specified with the design of the BE observer in Section~\ref{sec: case study}.

\subsection{Robust ECBF}\label{sec:robust ECBF sub}
The uncertain augmented system in \eqref{eq:aug_sys} based on the estimation of the environment state is rewritten as
\begin{equation}
\begin{split}
\dot\xi &= F(\xi)+G(\xi)\hat u_\xi +G_d d_{\dot e},
\end{split}
\label{eq:dis_augsys}
\end{equation}
where 
$G_d = \begin{bmatrix}{\bm 0} \\ I^{n_e}\end{bmatrix}$, $\hat u_\xi = \begin{bmatrix} u \\ \hat{\dot e}\end{bmatrix}$. Then a robust ECBF for \eqref{eq:dis_augsys} can be found based on Definition~\ref{def:ECBF for aug est system} and Theorem~\ref{theo:ECBF for aug est system} as below.
\begin{definition}\label{def:ECBF for aug est system}
A continuously differentiable function $H(\xi):\mathcal{X}\times E \rightarrow \mathbb{R}$ is a robust ECBF for system (\ref{eq:dis_augsys}) if there exists a class $\mathcal{K}$ function $\alpha$, such that $\forall \xi\in\mathcal{C}_\xi$
\begin{equation}\small
\begin{split}
    \sup_{u\in U}
L_FH(\xi)+L_GH(\xi)\hat u_\xi +\alpha(H(\xi))\geq \norm{L_{G_d}H(\xi)}\varepsilon_2.
\end{split}
\label{eq:RECBF}
\end{equation}
\end{definition}
\begin{theorem}\label{theo:ECBF for aug est system}
If $H(\xi)$ is a robust ECBF for (\ref{eq:dis_augsys}), then any locally Lipshitz continuous controller $u=K(\xi,\hat{\dot e})$ satisfying 
\begin{equation}
L_FH(\xi)+L_GH(\xi)\hat u_\xi +\alpha(H(\xi))\geq \norm{L_{G_d}H(\xi)}\varepsilon_2,
\label{eq:RECBF_condition}
\end{equation}
with $\forall \xi \in \mathcal{C_\xi}$ renders $\mathcal{C_\xi}$ forward invariant (safe), i.e., it ensures $\forall \xi_0\in \mathcal{C_\xi} \Rightarrow \xi \in \mathcal{C_\xi}, \forall t \geq 0$.
\end{theorem}
\begin{proof}
The proof can be found in~\cite{alan2021safe}.
\end{proof}

The robust ECBF based on Theorem~\ref{theo:ECBF for aug est system} only considers the influence of $d_{\dot{e}}$. In addition to this term, the measurement noise $d_{e}$ in the system (\ref{eq:dis_augsys}) also needs to be considered in the robust ECBF. To handle this problem, this paper designs a feedback controller, which is adapted from the method in~\cite{agrawal2022safe}, to guarantee the safety of the control system considering both $d_e$ and $d_{\dot e}$. This design requires that the functions $L_FH(\xi), L_GH(\xi), L_{G_d}H(\xi)$ and $\alpha(H(\xi))$ are Lipschitz continuous in argument $e$ on $\mathcal{C}_\xi$, i.e.,
\begin{equation}
\begin{split}
    L_FH(\xi)- L_FH(\hat\xi)
    &\geq
    -\mathcal{L}_{L_F}
    \norm{e-\hat e},\\	
    ||L_{G_d}H(\xi)||- ||L_{G_d}H(\hat\xi)||
    &\geq
    -\mathcal{L}_{L_{G_d}}
    \norm{e-\hat e},\\
    [L_{G}H(\xi)]_i- [L_{G}H(\hat\xi)]_i
    &\geq
    -\mathcal{L}_{L_{G_i}}
    \norm{e-\hat e},\\
    \alpha(H(\xi))-\alpha(H(\hat\xi))
    &\geq
    -\mathcal{L}_{\alpha\circ H}
    \norm{e-\hat e},
\end{split}
\label{eq:lipschitz}
\end{equation}
where $\mathcal{L}_{L_F}, \mathcal{L}_{L_G}, \mathcal{L}_{L_{G_d}}$ and $\mathcal{L}_{\alpha \circ H}$ are Lipschitz coefficients, $\hat \xi =\begin{bmatrix} x \\ \hat{e}\end{bmatrix}$, and $[L_{G}H(\xi)]_i$, $i \in \mathbb{I}_1^{n_u+n_e}$, refers to the $i$-th element of $L_{G}H(\xi)$. This leads to sufficient conditions for designing the robust ECBF for system~\eqref{eq:sys}, as shown in Theorem~\ref{theo:final ECBF}.
\begin{theorem}\label{theo:final ECBF}
If $H(\xi)$ is an ECBF for \eqref{eq:aug_sys}, and  $\varepsilon_1$ and $\varepsilon_2$ can be found via \eqref{eq:e_bound}, then a locally Lipschitz continuous controller 
\begin{equation}\small
\begin{split}
    k(\hat\xi, \hat{\dot e})= &\arg\min_{u\in\mathbb{R}^{n_u}} \norm{u-u_{\rm des}}^2\\
    & {\rm s.t.}~a(\hat \xi)+ \sum^{n_u+n_e}_{i=1} \min \{ b^{-}_i(\hat \xi)\hat u_{\xi_i}, b^{+}_i(\hat \xi)\hat u_{\xi_i} \}\geq 0
\end{split}
\label{eq:finalECBF}
\end{equation}
with
\begin{equation}\small
\begin{split}
    a(\hat\xi) &= L_FH(\hat\xi)-||L_{G_d}H(\hat\xi)||\varepsilon_2+\alpha(H(\hat\xi))-\Delta_a(\hat\xi),\\
    b^{-}_i &= [L_GH(\hat \xi)]_i- \Delta_{b_i}(\hat\xi), i \in \mathbb{I}_1^{n_u+n_e},\\
    b^{+}_i &= [L_GH(\hat \xi)]_i+ \Delta_{b_i}(\hat\xi), i \in \mathbb{I}_1^{n_u+n_e}, \\
    \Delta_a(\hat\xi) &= (\mathcal{L}_{L_F}+ \mathcal{L}_{L_{G_d}}\varepsilon_2+ \mathcal{L}_{\alpha \circ H})\varepsilon_1,\\
    \Delta_{b_i}(\hat\xi) &= \mathcal{L}_{L_{G_i}}\varepsilon_1, i \in \mathbb{I}_1^{n_u+n_e},
\end{split}
\end{equation}
renders $ \mathcal{C}_\xi$ forward invariant (safe) $\forall \hat \xi \in \mathcal{C}_\xi$ and $\forall  ||\hat{\dot e}||\leq \varepsilon_2$, i.e., it ensures $\forall \xi_0\in \mathcal{C}_\xi \Rightarrow \xi \in  \mathcal{C}_\xi, \forall t \geq 0$.
\end{theorem}
\begin{proof}
Firstly, with (\ref{eq:lipschitz}), there exists  
\begin{equation}
\begin{split}
    L_GH(\xi)\hat u_{\xi} &= \sum^{n_u+n_e}_{i=1}  [L_GH(\xi)]_i \hat u_{\xi_i}\\
&\geq 
\sum^{n_u+n_e}_{i=1}  ([L_GH(\hat\xi)]_{i}  -\mathcal{L}_{L_{G_i}}\norm{e-\hat e})\hat u_{\xi_i}\\
&\geq 
\sum^{n_u+n_e}_{i=1} 
\min \{ b^{-}_i(\hat \xi)\hat u_{\xi_i}, b^{+}_i(\hat \xi)\hat u_{\xi_i} \}.
\end{split}
\end{equation}

Note that here we suppose $\text{sign}(b^{-}_i(\hat \xi))=\text{sign}(b^{+}_i(\hat \xi))$ $\forall \hat\xi\in\mathcal{C}_\xi$ under the assumption that the estimation error can be sufficiently small. Secondly we show that with (\ref{eq:lipschitz}), the following condition holds that
\begin{equation}
\begin{split}
    &L_FH(\xi)-\norm{L_{G_d}H(\xi)}\varepsilon_2+\alpha(H(\xi))
        \geq \\
    &(L_FH(\hat\xi)-||L_{G_d}H(\hat\xi)||\varepsilon_2+\alpha(H(\hat\xi)))-\\
    &(\mathcal{L}_{L_F}+ \mathcal{L}_{L_{G_d}}\varepsilon_2+ \mathcal{L}_{\alpha \circ H})\norm{e-\hat e} \geq a(\hat \xi).
\end{split}
\end{equation}

With (\ref{eq:finalECBF}) satisfied, we can guarantee that $H(\xi)$ is a robust ECBF and by Theorem~\ref{theo:ECBF for aug est system}, the control system is safe.
\end{proof}
\begin{remark}
The application of the proposed method is not limited to single obstacle scenarios, as the observations of multiple obstacles can be directly incorporated into the augmented system \eqref{eq:dis_augsys}. This paper focuses on showcasing the method's ability to preserve robustness while reducing conservatism. In environments with multiple uncertain obstacles, challenges like coordinating between obstacles and ensuring computational efficiency have to be considered. Addressing these problems presents topics for future research.
\end{remark}

\section{Case Study}\label{sec: case study}
\begin{figure*}[t]
\SetFigLayout{4}{1}
\centering
\subfigure[Nominal ECBF.]{\includegraphics[width=1.35\columnwidth]{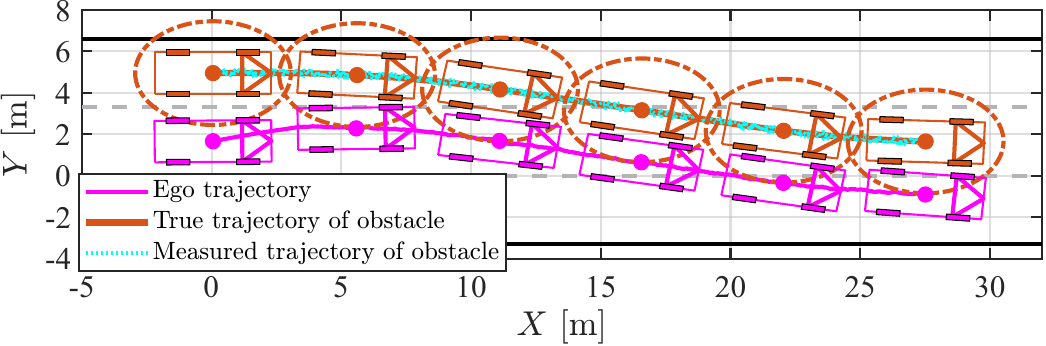}}
\hfill
\centering
\subfigure[Robust ECBF.]{\includegraphics[width=1.35\columnwidth]{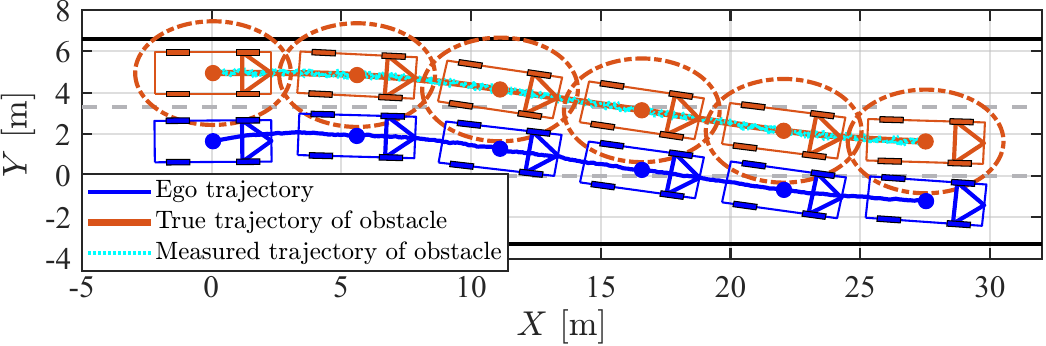}}
\hfill
\subfigure[Observer-based robust ECBF (proposed).]{\includegraphics[width=1.35\columnwidth]{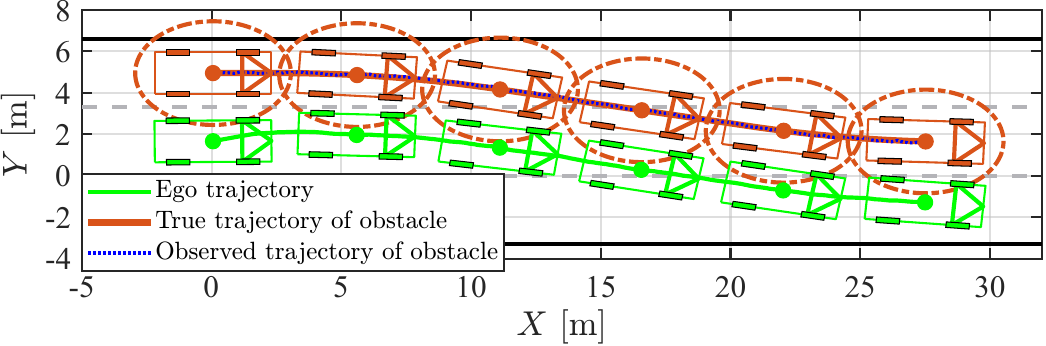}}
\caption{Trajectories generated by different methods in a dynamic environment with measurement uncertainties. The solid black line represents the road boundary, and the dotted grey line represents the road lane. The orange dotted-line circle is the unsafe set centered by the obstacle's center of geometry.}
\label{fig:traj}
\end{figure*}
\begin{figure}[t]
\centering
{\includegraphics[width=1\columnwidth]{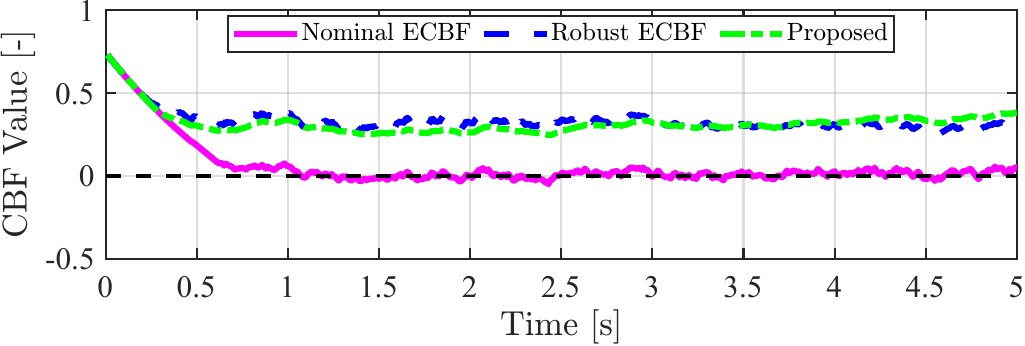}}
\caption{Value of CBF with different approaches.}
\label{fig:h}
\end{figure}
\begin{figure}[t]
\SetFigLayout{2}{1}
\centering
\subfigure[Slip angle $\beta$.]{\includegraphics[width=1\columnwidth]{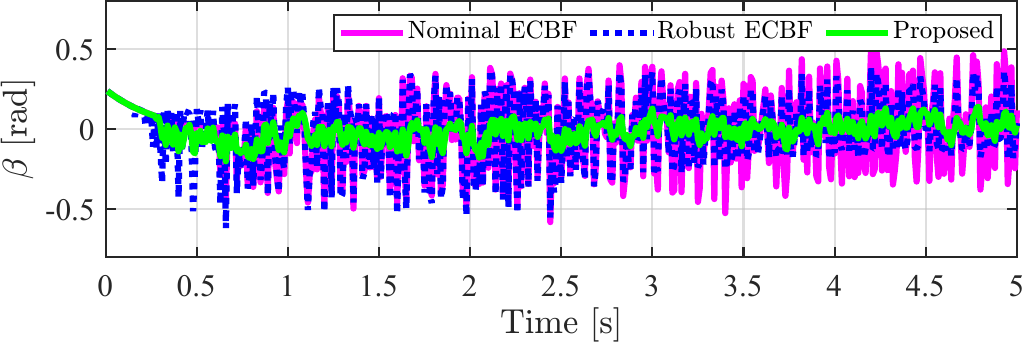}}
\hfill
\centering
\subfigure[Steering angle $\delta_f$.]{\includegraphics[width=1\columnwidth]{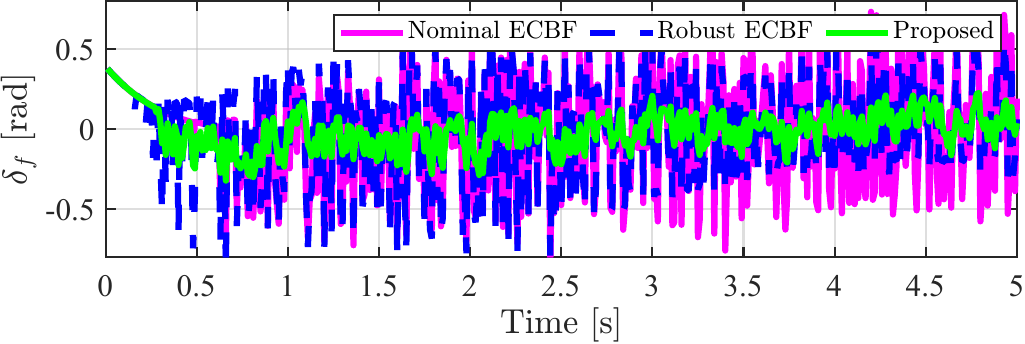}}
\caption{Control inputs generated by different approaches.}
\label{fig:input}
\end{figure}
\begin{figure}[htb]
\centering
{\includegraphics[width=0.99\columnwidth]{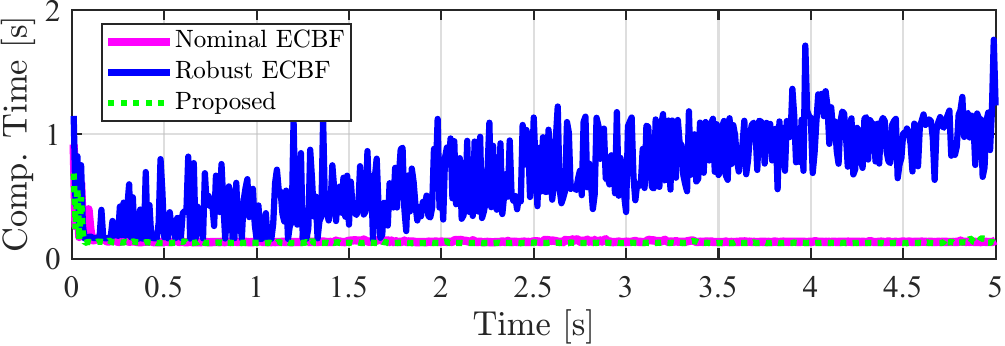}}
\caption{Computation time of different approaches}
\label{fig:comp}
\end{figure}
This section performs a specific case study to verify the performance of the proposed method. Section~\ref{sec:scenario and system modeling} describes the application scenario, the modeling of the control system and the obstacle, as well as the safe set for the controller. Then, the details for designing the BE observer in the case study are introduced in Section~\ref{sec: BE observer case study}. Finally, Section~\ref{sec:results discussion} presents the results and a discussion.

\subsection{Scenario Description and System Modeling}\label{sec:scenario and system modeling}
The case study considers an obstacle-avoidance scenario where the primary goal of the ego vehicle (control system) is to change lanes to the left, and at the same time, a surrounding vehicle (obstacle) in the left lane suddenly changes lanes to the right. The ego vehicle is modeled by a single-track kinematic model as
\begin{equation}
\begin{split}
\dot X &= v\cos(\psi+\beta),\\
\dot Y &= v\sin(\psi+\beta),\\
\dot \psi &= \frac{v}{lr}\sin(\beta),\\
\dot v &= a,\\
\beta &= {\rm arctan}\left( \frac{l_r}{l_f+l_r}\tan(\delta_f) \right),
\end{split}
\label{eq:nonaffine_sys}
\end{equation}
where $a$ is the acceleration at the vehicle's center of gravity, and $\delta_f$ is the front steering angle, they are the inputs of the system. Variable $\psi$ is the orientation angle of the vehicle in the ground coordinate system, $v$ is the speed of the vehicle in the vehicle frame. Parameters $l_f$ and $l_r$ describe the distance from the vehicle's center of gravity to the front and rear axles, respectively, and $\beta$ is the slip angle of the vehicle. 

To reformulate \eqref{eq:nonaffine_sys} as a nonlinear affine system as in \eqref{eq:sys}, this paper assumes that the slip angle $\beta$ in the lane-changing scenario is sufficiently small~\cite{zhou2023interaction}. Therefore it holds that $\cos(\beta) \approx 1$ and $\sin(\beta)\approx \beta$. Based on this approximation, the model describing the ego vehicle in the form of \eqref{eq:sys} is
\begin{align}
\underbrace{
\begin{bmatrix}
\dot{X}\\
\dot{Y}\\
\dot{\psi}\\
\dot{v}
\end{bmatrix}}_{\dot{x}}=
\underbrace{
\begin{bmatrix}
v\cos(\psi)\\
v\sin(\psi)\\
0\\
0
\end{bmatrix}}_{f(x)}+
\underbrace{
\begin{bmatrix}
0 & -v\sin(\psi)\\
0 & v\cos(\psi)\\
0 & v/l_r\\
1 & 0
\end{bmatrix}}_{g(x)}
\underbrace{
\begin{bmatrix}
a\\
\beta
\end{bmatrix}}_{u}.
\label{eq:affine_sys}
\end{align}

System \eqref{eq:affine_sys} with inputs $a$ and $\beta$ will be used for the design of the robust ECBF of the ego vehicle. Note that the direct control inputs of the vehicle system should be $a$ and $\delta_f$, where $\delta_f$ can be calculated using $\beta$ through (\ref{eq:nonaffine_sys}).

The surrounding vehicle is described by the original nonlinear kinematic model (\ref{eq:nonaffine_sys}). Define the longitudinal and lateral positions of the surrounding vehicle in the ground coordinate system as $(X^s, Y^s)$. Then, the safe set for collision avoidance between the ego vehicle and the obstacle can be formulated as
$$\mathcal{C}_\xi = \{\xi: a^s(X-X^s)^2+b^s(Y-Y^s)^2-1 \geq 0\},$$
where $a^s=\frac{1}{r^2_a}$ and $b^s=\frac{1}{r^2_b}$, $r_a$ and $r_b$ denote the semi-major axis and semi-minor axis of the ellipsoidal unsafe set. In the lane-change scenario, considering that the heading angle of the surrounding vehicle is relatively small, the parameters $r_a$ and $r_b$ can be chosen with sufficient margin to cover the variations of the orientation of the surrounding vehicle, such that the size of the unsafe set remains constant. Since the ego vehicle has to avoid collision with the position of the surrounding vehicle, the environmental state in the robust ECBF is defined as $e = [{X}^s \ {Y}^s]^{\rm T}$.

\subsection{Bounded-Error Observer for Robust ECBF}\label{sec: BE observer case study}
The BE observer is significant for designing the controller in Theorem~\ref{theo:final ECBF}. To design the observer \eqref{eq:observer} that satisfies \eqref{eq:BE_con}, this paper adopts a nonlinear observer from~\cite{jeon2019tracking} to estimate the state of the surrounding vehicle. The observer is designed by solving a semi-definite programming (SDP) problem
\begin{equation}\small \label{eq:SDP}
\begin{aligned}
\mathop{\rm minimize}\limits_{P, \ R} \quad & \gamma & \\
{\rm subject \ to}\quad
&{\bm 0} \preceq P -I, \ {\bm 0} \preceq
\begin{bmatrix}
    P&R^T\\R&\gamma I
\end{bmatrix},&\\
&A^TP+PA-C^TR-R^TC+2\theta P \preceq {\bm 0}, & \\
&\forall u^s\in u^s_{\rm grid}, \ \forall x^s\in x^s_{\rm grid},&
\end{aligned}
\end{equation}
\noindent where $ C =  c(x^s)$, $ D =  c_v(x^s)$, $A = \frac{\partial}{\partial \hat x^s}(f^s(\hat x^s, u^s))$, $\theta\geq 0$ is a design parameter, $\{x^s,u^s\}_{\rm grid}$ are the operating ranges of the surrounding vehicle containing finite elements, such that the number of constraints of \eqref{eq:SDP} is limited. From the solution of \eqref{eq:SDP}, the observer gain can be obtained as $L = P^{-1}R^T$ such that (\ref{eq:observer}) is represented as $\dot{\hat x}^s = f^s(\hat x^s, u^s)+L(y^s-C\hat x^s)$. To find the bounded set for the estimation error, i.e., $\mathcal{P}(t,\hat x^s)$, we additionally add the following linear matrix inequality (LMI)-based $H_\infty$ condition adopted from~\cite{zemouche2008observers} to the SDP problem (\ref{eq:SDP}):
\begin{equation}\small
\begin{split}
    \begin{bmatrix}
        A^TP+PA-C^TR-R^TC+I^{n_x} &-R^TD\\
        -D^TR & -\lambda^2I^{n_x}
    \end{bmatrix}\preceq {\bm 0},
\end{split}
\end{equation}
where $\lambda>0$ is a prescribed scalar disturbance attenuation level. 
With this LMI condition, $\mathcal{P}(t,\hat x^s)$ can be found as~\cite{zemouche2008observers}:
\begin{equation}
\mathcal{P}(t,\hat x^s) = \mathcal{P}(\hat x^s) = \{x^s\in {\mathcal{X}^s}: \norm{x^s-\hat x^s} \leq \lambda \bar w  \}.
\end{equation}
\begin{remark}
To stabilize the ego vehicle while tracking the desired speed $v_d$ and the goal state, instead of designing the desired control signal $u_{\rm des}$, the following control Lyapunov function (CLF)-based constraints~\cite{he2021rule} are added to \eqref{eq:finalECBF}:
\begin{equation}
L_fV_j(x)+L_gV_j(x)u \leq -\alpha_{v_j} V_j(x)+\rho_j,
\end{equation}
where $V_j:\mathcal{X}\rightarrow\mathbb{R}_{+}$ is a CLF, $\alpha_{v_j}$ is a class $\mathcal{K}$ function and $\rho_j$ is a relaxation variable with $j$ representing subscript $v$, $y$ or $\psi$. The CLF $V_j(x)$ is designed as
\begin{equation}
\begin{split}
    V_v(x) &= (v-v_d)^2,
    V_y(x) = (Y-Y_l)^2,
    V_\psi(x) = \psi^2,
\end{split}
\end{equation}
where $Y_l$ is the lateral coordinate of the center line of the target lane of the ego vehicle. 
\end{remark}

\subsection{Results and Discussion}\label{sec:results discussion}
In the simulation scenario, the measurement noise to the surrounding vehicle's position and velocity are designed as $\bar w=0.2 \ {\rm m}$ and $\bar d = 0.2\ {\rm m/s}$, respectively. The disturbance attenuation level in the observer is $\lambda = 0.8$. Both the ego vehicle and the obstacle drive with constant speeds during the simulation, where the techniques of CBF-based longitudinal and lateral control are discussed in~\cite{he2021rule}. The optimization problems are solved by YALMIP~\cite{lofberg2004yalmip} with MATLAB 2023b.

In the case study, the proposed method is compared with a nominal ECBF and a robust ECBF method in~\cite{molnar2022safety}. The trajectories of the ego vehicle generated by each method with respect to the true and measured trajectories of the obstacle are shown in Fig.~\ref{fig:traj}. Furthermore, the value of CBF, i.e., the value of function $H$ in~\eqref{eq:C_xi}, is compared in Fig.~\ref{fig:h}. The control inputs $\beta$ and $\delta_f$ of model \eqref{eq:affine_sys} by each method are shown in Fig.~\ref{fig:input}, and the computation time is compared in Fig.~\ref{fig:comp}. It is seen in Fig.~\ref{fig:traj} that the ego vehicle with the nominal ECBF is closer to the obstacle, while the robust ECBF and the proposed method keep a larger distance between the ego vehicle and the obstacle. This is consistent with the results in Fig.~\ref{fig:h}, where the CBF value of both robust ECBF and the proposed method is always positive, this implies that the safety constraint with respect to an uncertain obstacle can be satisfied by considering the uncertainties in the controller. Although the safety is guaranteed in both the robust ECBF and the proposed method, Fig.~\ref{fig:input} shows that the magnitude of control inputs by the proposed method is smaller. This is because the proposed method utilizes the observed environment state instead of the directly measured state that has comparatively larger uncertainties. In addition, the proposed method based on solving a QP problem also outperforms the robust ECBF, which solves an SOCP problem, in terms of the computation time, as shown in Fig.~\ref{fig:comp}. 

\section{Conclusion}
This paper proposes an observer-based safety-critical controller in dynamic environments using environmental control barrier functions. The proposed method is designed to be robust against the measurement uncertainties of moving obstacles in a dynamic environment. The simulation results of the collision-avoidance problem of an autonomous ego vehicle with an uncertain surrounding vehicle show that (1) The proposed method is safer than the nominal ECBF due to the consideration of environmental uncertainties, and (2) The method reduces conservatism and increases computational efficiency compared with the robust ECBF by applying the state observer of the obstacle system.

\end{document}